\documentclass[a4paper,12pt]{article}
\usepackage{amssymb,amsfonts,amsmath,mathtext,cite,enumerate,float}
\usepackage{amssymb,amsthm,latexsym,euscript,amscd, authblk}
\usepackage{braket}

\usepackage[english]{babel}
\usepackage[utf8]{inputenc}

\newtheorem{lemma}{Lemma}
\newtheorem{theorem}{Theorem}

\title {On Definition of Quantum Tomography via the Sobolev Embedding Theorem}

\author[1]{Grigori Amosov\thanks {gramos@mi-ras.ru}}

\author[2]{Yakov Korennoy\thanks {abc772211@mail.ru}}

\affil[1]{Steklov Mathematical Institute of Russian Academy of Sciences, 8 Gubkina St., Moscow, 119991 Russia}

\affil[2]{Lebedev Physical Institute of Russian Academy of Sciences,
Leninskii prospect 53, Moscow 119991, Russia}

\begin{document}

\maketitle

\begin{abstract}
We obtain sufficient conditions on kernels of quantum states under which Wigner functions, optical quantum tomograms and linking their formulas are correctly defined. Our approach is based upon the Sobolev embedding theorem. The transition probability formula and the fractional Fourier transform are discussed in this framework.
\end{abstract}



\section{Introduction}

In \cite{MankoPhysLettA96}  it was shown that the states of
$n$-dimensional quantum systems can be completely described by the
real and positive probability distribution functions of physical
observable $\overline{\hat x}(\overline\alpha)=\overline{\hat
q\cos\alpha}+\overline{\hat p\sin\alpha}$, where $\overline{\hat
q\cos\alpha}$ and $\overline{\hat p\sin\alpha}$ are $n$-dimensional
vectors with the components $\hat q_j\cos\alpha_j$ and $\hat
q_j\sin\alpha_j$. Here $\overline{\hat{q}}$ and $\overline{\hat{p}}$
are position and momentum operators, $\overline\alpha$ are angular
parameters. Such a representation of quantum states  is known as
tomographic. In a variety of subsequent articles the properties of
tomographic representation were investigated in detail (see, e.g.
Review \cite{IbortPhysScr}).

The ordinary definition of optical quantum tomogram reads \cite{MankoPhysLettA96}
\begin{equation}\label{tom}
\omega_{\hat\rho}(\overline x,\overline\alpha)=
\mathrm{Tr}\left\{\hat\rho\delta(\overline x\hat I
-\overline {\hat q\cos \alpha}-\overline{\hat p\sin\alpha})\right\},
\end{equation}
where $\hat\rho$ is a density operator of the state.

Let us make sure that a trace in the righthand side of (\ref {tom}) is correctly defined
for all $\hat\rho>0$ with $\mathrm{Tr}\{\hat\rho\}=1$ and the kernels $\rho(q,q')\in L_2({\mathbb R}^{2n})$.
Following to \cite {Hol} let us define the characteristic function of a state $\hat \rho$ by the formula
\begin{equation}\label{chF}
f_{\hat \rho}(\overline x,\overline y)=\mathrm{Tr}\left\{\hat \rho e^{i(\overline x\cdot\overline {\hat q}
+\overline y\cdot\overline {\hat p})}\right\}.
\end{equation}
It is known \cite{Hol} that $f_{\hat \rho}\in L_2({\mathbb R}^{2n})\cap C({\mathbb R}^{2n})$. Hence we can define a function
\begin{equation*}
F_{\hat \rho}(\overline t,\overline \alpha )=f_{\hat \rho}(\overline{t\cos\alpha},\overline{t\sin\alpha}).
\end{equation*}
It follows from the inclusion $f_{\hat \rho}\in L_2({\mathbb R}^{2n})$ that functions
$$
f_{\overline s,\overline \alpha}(\overline t)=f_{\hat \rho }(\overline s+\overline {t\cos \alpha } ,\overline s+\overline {t\sin\alpha })\in L_2({\mathbb R}^n)
$$
for almost all $\overline s\in {\mathbb R}^n$. Hence $F_{\hat \rho }(\cdot ,\overline \alpha )\in L_2({\mathbb R}^n)$ due to $f_{\hat \rho}$ is continuous.
Formula (\ref {tom}) means that
$\omega _{\hat \rho }(\cdot,\overline \alpha)$ is the Fourier transform of $F_{\hat \rho }(\cdot,\overline \alpha )$
and we see that it is correct as the Fourier transform of $L_2$-function.

It should be noted that
given a characteristic function $f_{\hat \rho}$ it is possible to regenerate $\hat \rho $ in weak sense. Then, for a kernel of $\hat \rho $
in the coordinate representation we get \cite{Hol}
\begin{equation}\label{weak}
\rho (\overline q,\overline q')=\frac {1}{(2\pi )^n}\int \limits _{{\mathbb R}^{n}}e^{-\frac {i}{2}(\overline q+\overline q')\overline y}f_{\hat \rho }(\overline q-\overline q',\overline y)d\overline y.
\end{equation}
Thus, (\ref {weak}) is the partial Fourier transform of $f_{\rho }$.
Hence we can find
kernels $\rho _{\overline \alpha }(\cdot ,\cdot )\in L^2({\mathbb R}^{2n})$ of the operator $\hat \rho$ in integral
representations associated with all the observables $\overline {\hat x}(\overline \alpha )$. It results in a tomogram can
be correctly defined by the formula
\begin{equation}\label{Tom}
\omega _{\hat \rho }(\overline x,\overline \alpha )=\rho _{\overline \alpha }(\overline x,\overline x)
\end{equation}
for any state $\hat \rho $, and $\omega _{\hat \rho}(\cdot ,\overline \alpha )\in L_1({\mathbb R}^n)$. Thus,
formulae (\ref {tom}) and (\ref {Tom}) give the same result for all states $\hat \rho $. Note that to take a trace in (\ref {Tom}) we take into account that $\hat \rho >0$. For an arbitrary $\hat \rho$ of the trace class (\ref {Tom}) is not valid. In our work, we will waive the requirement of positivity for $\hat \rho$.

The situation becomes more complex for the Wigner function
\begin{equation}\label{Wigner}
\mathcal {W}_{\hat \rho}(\overline q,\overline p)=\frac {1}{(2\pi )^n}\int \limits _{{\mathbb R}^{2n}}e^{i\overline q\cdot\overline x
+i\overline p\cdot\overline y}f_{\hat \rho }(\overline x,\overline y)d^n\overline xd^n\overline y.
\end{equation}
The claim $f_{\hat \rho }(\cdot , \cdot)\in L_2({\mathbb R}^{2n})$ results in $\mathcal {W}_{\hat \rho}(\cdot ,\cdot )\in L_2({\mathbb R}^{2n})$ but taking the Fourier transform we lose smoothness. Thus, the restriction of $\mathcal {W}_{\hat \rho }$ to a fixed hyperplane doesn't exist in general. This prevents the use of the Radon transform for $\mathcal {W}_{\hat \rho }$.
To avoid these difficulties, articles
concerning quantum tomography often deal only with a variety of bound states of quantum systems. In the case, wave functions
as well as density matrixes of such states belong to the Schwartz spaces, i.e.
$\psi(\cdot )\in \mathcal{S}(\mathbb{R}^n)$ and $\rho(\cdot,\cdot)\in \mathcal{S}(\mathbb{R}^{2n})$.

Up to the present time the question is opened about what class of functions do we need for the integral formulae connecting quantum tomograms and Wigner functions to be correct. In the present paper we will partially fill this gap using the famous Sobolev Embedding Theorem.

\section {Preliminaries}

Following \cite{RS} denote $\mathcal {S}({\mathbb R}^n)$ the Schwartz space consisting of
infinitely differentiable and fast descending at infinity functions $\psi (\overline x)$ of $n$ variables
$\overline x=(x_1,\dots ,x_n)$. All functions $\psi \in \mathcal{S}({\mathbb R}^n)$ are known to be summable
with respect to any choice of variables $x_{j_1},\dots ,x_{j_k}$. Moreover,
$$
\int \limits _{{\mathbb R}^k}\psi (\overline x)dx_{j_1}\dots dx_{j_k}\in \mathcal{S}({\mathbb R}^{n-k}).
$$
The Fourier transform is correctly defined for $\psi \in \mathcal{S}({\mathbb R}^n)$ by the formula
\begin{equation}\label{F}
{\mathcal F}[\psi ](\overline x)=\frac {1}{(2\pi )^n}\int \limits _{\mathbb {R}^n}e^{-i\overline x\cdot\overline y}
\psi (\overline y)d^n\overline y.
\end{equation}
It is known that ${\mathcal F}\left[\mathcal{S}({\mathbb R}^n)\right]=\mathcal{S}({\mathbb R}^n)$.
Due to the Plancherel equality
$$
(2\pi )^n\int \limits _{\mathbb {R}^n}|{\mathcal F}[\psi](\overline x)|^2d^n\overline x=\int \limits _{\mathbb {R}^n}|
\psi (\overline y)|^2d^n\overline y
$$
the Fourier transform (\ref {F}) can be extended to $L_2({\mathbb R}^n)$.
Fix the set of indices $J=\{j_1,\dots ,j_k\}$, $0<k<n$, then  also the partial Fourier transform ${\mathcal F}_J$
can be correctly defined
on $\mathcal{S}({\mathbb R}^n)$ as follows
\begin{equation}\label{PF}
{\mathcal F}_J[\psi](\overline x)=\frac {1}{(2\pi)^k}\int \limits _{\mathbb {R}^k}e^{-i\sum \limits _{j\in J}x_jy_j}
\psi (\overline y)\prod _{j\in J} dy_{j}
\Bigg\vert_{\{y_l=x_l,\,\,l\in\{1,...,n\}\backslash J\}}\ \ .
\end{equation}
Analogously to ${\mathcal F}$ taking into account the Plancherel equality we get that (\ref {PF}) can be extended
to the space $L_2({\mathbb R}^k)$.

Given $\nu >0$ denote $W_2^{\nu }({\mathbb R}^n)$ the Sobolev space consisting of functions $\psi \in L_2({\mathbb R}^n)$ such that
the function $\tilde \psi (\overline x)=|\overline x|^{\nu}{\mathcal F}[\psi](\overline x)$ is square summable.
If $\nu =k$ is an integer number, $W_2^{k}({\mathbb R}^n)$ consists of functions
possessing $k$ weak derivatives from $L_2({\mathbb R}^n)$ in any variables. Let us consider the spaces
\begin{equation}\label{space}
{{\mathcal V}({\mathbb R}^{2n})}=W_2^{n+1}({\mathbb R}^{2n})\cap {\mathcal F}\left[W_2^{n+1}({\mathbb R}^{2n})\right],
\end{equation}
\begin{equation*}
{\mathbb V}({\mathbb R}^{n})=W_2^{n+1}({\mathbb R}^{n})\cap {\mathcal F}\left[W_2^{n+1}({\mathbb R}^{n})\right].
\end{equation*}

\begin {lemma}\label{1} The spaces ${{\mathcal V}({\mathbb R}^{2n})}$ and ${\mathbb V}({\mathbb R}^n)$ are invariant
with respect to the Fourier
transform and the partial Fourier transform as well.
\end{lemma}

\begin{proof}

The operator of multiplication by the variable $x_j$ is mapped to the differentiation $-i\frac {\partial }{\partial x_j}$
and vice versa. The result follows.

\end{proof}

Given a function $f\in L_2({\mathbb R}^n)$ one can try to define a function $F\in L^2({\mathbb R}^m),\ m<n,$
which is a restriction of $F=f|_S$ to some hyperplane $S$ with codimension $n-m$. Such the restriction $F$ known as a trace of
$f$ is not defined in general.
Our consideration is based upon the following famous statement \cite {Sobolev, Nikolski}.

{\bf The Sobolev Embedding Theorem.} {\it Suppose that $f\in W_2^{\nu}({\mathbb R}^N)$, then
$$
(\mathrm{i})\ \ \ f\in C^{\nu-[\frac {N}{2}]-1}({\mathbb R}^N)
$$
(\cite {Sobolev}).
Moreover, the trace $F=f|_{S}$ to any hyperplane $S$ of codimension $N-m$
exists and
$$
(\mathrm{ii})\ \ \ F\in W_2^{\nu-\frac {N-m}{2}}({\mathbb R}^m)
$$
(Theorem 2 in \cite {Nikolski}).
}

The lemma below shows why the space ${\mathcal V}({\mathbb R}^n)$ can be useful in integral transformations.

\begin {lemma}\label {2} Given a function $f(x_1,\dots x_n,y_1,\dots y_2)$ belonging to ${\mathcal V}({\mathbb R}^{2n})$ and a vector
$\overline \alpha \in [0,2\pi]^n$
the trace
$$
F_{\overline\alpha}(t_{1},\dots ,t_{n})=f(t_1\cos\alpha_1,\dots,t_n\cos\alpha_n,t_1\sin\alpha_1,\dots,t_n\sin\alpha_n)
$$
of $f$ to the hyperplane $x_j\sin\alpha_j-y_j\cos\alpha_j=0,\ 1\le j\le n,$ is correctly defined and
$F_{\overline \alpha}\in C({\mathbb R}^n)\cap L_1({\mathbb R}^n)$.
\end{lemma}

\begin{proof}

Substituting $N=2n$, $\nu =n+1$ and $m=n$ to Sobolev Embedding Theorem we obtain that $F_{\overline\alpha}$ is correctly defined and
the inclusion $F_{\overline\alpha}\in W_2^{\frac {n}{2}+1}({\mathbb R}^n)$ holds due to (ii). 
Then, (i) results in $F_{\overline\alpha}\in C({\mathbb R}^n)$.
Since $f\in {\mathcal V}({\mathbb R}^{2n})$ we get ${\mathcal F}[f]\in W_2^{n+1}({\mathbb R}^{2n})$.
It gives rise $|\overline x|^{n+1}f(\overline x)$ lies in $L_2({\mathbb R}^{2n})$. 
Hence, $|\overline t|^{n+1}F_{\overline\alpha}(\overline t)$
belongs to $L_2({\mathbb R}^n)$. Applying
the Schwartz inequality
\begin{multline}\label{sch}
\left |\int \limits _{{\mathbb R}^n}F_{\overline\alpha}(t_{1},\dots ,t_n )dt_{1}\dots
dt_{n}\right |^2  \\
\le\left  (\int \limits _{{\mathbb R}^n}\prod
_{j=1}^n\frac {1}{t_j^2+1}dt_{j}\right ) 
\left (\int \limits
_{{\mathbb R}^n}|F_{\overline\alpha}(t_1,\dots ,t_{n})|^2\prod
_{j=1}^n(t_{j}^2+1)dt_{j}\right )
\end{multline}
we obtain $F_{\overline\alpha}\in L_1({\mathbb R}^n)$.

\end{proof}

Take a function $\rho (\cdot ,\cdot)\in L_2({\mathbb R}^{2n})$ and consider the integral operator $\hat \rho$ defined by the formula
\begin{equation}\label{density}
\hat\rho[\psi](\overline x)=\int \limits _{{\mathbb R}^n}\rho (\overline x,\overline y)\psi (\overline y)d^n\overline y,
\ \psi \in L^2({\mathbb R}^n).
\end{equation}
If $\hat \rho $ belongs to the convex set $\mathfrak {S}$ consisting of hermitian and positive unit trace operators in
$L_2({\mathbb R}^n)$, it is called {\it a quantum state}.
The important subclass of $\mathfrak {S}$ is
pure quantum states $\hat \rho $ with the kernels $\rho (\overline x,\overline y)=\xi (\overline x)\xi ^*(\overline y)$,
$\xi \in L_2({\mathbb R}^n), \|\xi \|=1$.
The characteristic function $f_{\hat \rho}\equiv f_{\rho}$ of $\hat \rho\in \mathfrak {S}$ is correctly defined by the formula (\ref{chF}).
For a pure quantum state $\hat \rho =|\xi \rangle\langle\xi |$ we obtain
\begin{equation*}
f_{\xi }(\overline x,\overline y)=\langle\xi |e^{i\overline x\cdot\overline{\hat q}+i\overline y\cdot\overline{\hat p}}|\xi \rangle.
\end{equation*}

Using the Baker formula $e^{ix\hat q_j+iy\hat p_j}=e^{ix\hat q_j}e^{iy\hat p_j}e^{\frac {ixy}{2}}$ we can rewrite
(\ref{chF}) as follows
\begin{equation}\label{Baker}
f_{\rho}(\overline x,\overline y)=\int \limits _{{\mathbb R}^n}e^{i\overline x\cdot\overline t}\rho \left (\overline t
+\frac {\overline y}{2},\overline t-\frac {\overline y}{2}
\right )d^n\overline t.
\end{equation}
It follows from (\ref {Baker}) that the map $\hat \rho \to f_{\rho}$ can be extended to the Hilbert-Schmidt operators
$\hat \rho$ having the kernels $\rho (\cdot ,\cdot )\in L_2({\mathbb R}^{2n})$. In the case,
$f_{\rho }(\cdot ,\cdot)\in L_2({\mathbb R}^{2n})$ \cite {Hol}.

\begin {lemma}\label{3} Formula (\ref{Baker}) defines a linear map on the space ${\mathcal V}({\mathbb R}^{2n})$. The invers
transformation is given by the formula
\begin{equation}\label{inverseBaker}
\rho(\overline q,\overline q')=\frac{1}{(2\pi)^n}\int \limits _{{\mathbb R}^n}e^{-i\overline x\cdot\overline t}
f_{\rho}(\overline x,\overline y)d^n\overline x
\Big|_{\overline t=(\overline q+\overline q')/2,~\overline y=\overline q-\overline q'}~~.
\end{equation}
\end{lemma}

\begin{proof}

The nonsingular change of variables $\overline q=\overline t+\overline y/2$,
$\overline{q'}=\overline t-\overline y/2$ in the function $\rho(\overline q,\overline{q'})$
as well as the Fourier transform with respect to one of coordinates map ${\mathcal V}({\mathbb R}^{2n})$
to itself. Applying the inverse Fourier transform to $f_{\hat \rho}(\overline x,\overline y)$ with respect to
$\overline x$ we obtain (\ref {inverseBaker}).
\end{proof}

\section {Characteristic function, Wigner function and optical quantum tomogram}

Consider an integral operator $\hat \rho$ with the kernel $\rho (\cdot ,\cdot )\in {\mathcal V}({\mathbb R}^{2n})$.Then,
Lemmas 1 and 3 result in the Wigner function (\ref {Wigner}) $\mathcal {W}_{\hat \rho}\equiv \mathcal {W}_{\rho}(\cdot ,\cdot )\in {\mathcal V}({\mathbb R}^{2n})$.

Obviously \cite {MankoPhysLettA96} the optical quantum tomogram is defined via the restriction of the partial
Fourier transform of the characteristic
function or the Radon transform \cite {Radon17} of the Wigner function. Nevertheless it  can not be correctly done
for an arbitrary $\rho (\cdot ,\cdot )\in L_2({\mathbb R}^{2n})$.
Nevertheless, the application of Lemmas 2 and 3 shows that it is correct for kernels $\rho (\cdot ,\cdot )\in {\mathcal V}({\mathbb R}^{2n})$.

\begin {theorem}\label{1} If a kernel $\rho (\cdot ,\cdot )\in {\mathcal V}({\mathbb R}^{2n})$, then
$f_{\rho}(\cdot )\in {\mathcal V}({\mathbb R}^{2n})$ and $\mathcal {W}_{\rho}(\cdot ,\cdot )\in  {\mathcal V}({\mathbb R}^{2n})$ such that
the following formulas are correct and define the same object called the optical quantum tomogram
\begin{equation}\label{tomogram}
\omega _{\rho}(\overline x,\overline \alpha )=\frac {1}{(2\pi
)^n}\int \limits _{{\mathbb R}^n} e^{-i\overline x\cdot\overline
t}f_{\rho }(\overline{t\cos \alpha }, \overline{t\sin \alpha
})d^n\overline t,
\end{equation}
\begin{equation}\label{tomogram2}
\omega _{\rho}(\overline x,\overline \alpha )=\frac{1}{(2\pi)^n}\int \limits _{{\mathbb R}^{2n}}\mathcal {W}_{\rho}(\overline q,\overline p)
\delta(\overline x-\overline{q\cos \alpha }-\overline{p\sin \alpha })d^n\overline qd^n\overline p,
\end{equation}
where  $\alpha _j\in [0,2\pi]$, $j\in \{1,\dots ,n\}$. Moreover,
$\omega _{\rho}(\cdot ,\overline \alpha )\in C({\mathbb R}^n)\cap L_1({\mathbb R}^n)$.
\end{theorem}

\begin{proof}

It follows from Lemma 1 that the characteristic function $f_{\rho}\in {\mathcal V}({\mathbb R}^{2n})$.
Let us make a change of variables
$$
\tilde x_j=x_j\cos\alpha _j+y_j\sin\alpha _j,\ \tilde y_j=x_j\sin\alpha _j-y_j\cos\alpha _j,\ 1\le j\le n.
$$
Then, take the partial Fourier transform ${\mathcal F}_{\overline {\tilde x}}$ of $f_{\rho}(\overline {\tilde x},\overline {\tilde y})$ with respect
to variables $\overline {\tilde x}$. Due to Lemma 1 ${\mathcal F}_{\tilde x}(f)\in {\mathcal V}({\mathbb R}^{2n})$.
Applying Lemma 2 we obtain that the trace $F_{\overline \alpha }$ of the function ${\mathcal F}_{\overline {\tilde x}}(f)$
with respect to the hyperplane $\tilde y_j=-x_j\sin\alpha_j+y_j\cos\alpha_j=0,\ 1\le j\le n$, is correctly
defined and $F_{\overline \alpha}=\omega _{\rho}(\cdot ,\overline \alpha )\in C({\mathbb R}^n)\cap L_1({\mathbb R}^n)$.
Since the Wigner function is the Fourier transform of the characteristic function we can conclude that
$\mathcal {W}_{\rho}(\cdot ,\cdot )\in  {\mathcal V}({\mathbb R}^{2n})$ due to Lemma 1.
Taking into account Lemma 2 we obtain that the traces of $\mathcal {W}_{\rho}$ are determined for any hyperplane
$-x_j\sin\alpha _j+y_j\cos\alpha_j=0,\ 1\le j\le n$. Applying
the Radon transform to $\mathcal {W}_\rho(\overline q,\overline p)$ we obtain (\ref {tomogram2}). The coincidence
of (\ref {tomogram}) and (\ref {tomogram2}) follows from the Fourier slice theorem \cite {Helgason}.

\end{proof}

{\bf Remark.} {\it Transformation (\ref{tomogram}) is reversible with the inverse Fourier transform
\begin{equation}        \label{omegatof}
f_\rho(\overline{\lambda\cos\alpha},\overline{\lambda\sin\alpha})=
\int\limits_{\mathbb R^n}e^{i\overline x\cdot\overline\lambda}
\omega _{\rho}(\overline x,\overline\alpha)d^n\overline x,
\end{equation}
and the change of variables from $\{(\overline t,\overline\alpha)\}$ to the Cartesian coordinates
$$
\lambda_j=\mathrm{sgn}(y_j)\sqrt{x_j^2+y_j^2},~~~~~\alpha_j=\cot^{-1}\frac{x_j}{y_j},~~~~j\in\{1,...,n\}
$$
gives us the characteristic function $f_\rho(\overline x,\overline y)$,
which can be transformed to the density matrix $\rho(\overline q,\overline q')$
with the help of formula (\ref{inverseBaker}).} $\Box $

\section {Transition probability between two states}

The transition probability $P_{12}$ between two states $\hat\rho_1$ and $\hat\rho_2$ of quantum system reads
\begin{equation}        \label{transprob}
P_{12}=\int_{\mathbb{R}^{2n}}\rho_1(\overline q,\overline{q'})
\rho_2(\overline{q'},\overline{q})d^n\overline qd^n\overline{q'}.
\end{equation}
The following theorem is valid.

\begin {theorem}\label{2} Given two optical tomograms $\omega _{\rho_1}(\overline x,\overline\alpha)$ and
$\omega _{\rho_2}(\overline x,\overline\alpha)$ corresponding to two density matrixes
$\rho_1(\overline q,\overline{q'})\in{\mathcal V}({\mathbb R}^{2n})$
and $\rho_2(\overline q,\overline{q'})\in{\mathcal V}({\mathbb R}^{2n})$  by means of formulas
(\ref{Baker}) and (\ref{tomogram}) we get
\begin{align}\label{equalityt2}
&P_{12}=\int\limits_{[0;\,\pi]^n\times\mathbb{R}^{n}} d^n\overline\alpha d^n\overline\lambda\,
\bigg(\prod_{j=1}^n|\lambda_j|\bigg)
\int\limits_{\mathbb{R}^{2n}}d^n\overline x d^n\overline{x'}e^{i\overline\lambda\cdot(\overline x-\overline{x'})}
\omega _{\rho_1}(\overline x,\overline\alpha)\omega _{\rho_2}(\overline{x'},\overline\alpha).
\end{align}
\end{theorem}

\begin{proof}

Let us change the variables $\overline q=\overline t+\overline y/2$,
$\overline{q'}=\overline t-\overline y/2$ in the integral (\ref{transprob})
and let us do  the partial Fourier transforms of functions $\rho_1(\overline t+\overline y/2,\overline t-\overline y/2)$,
$\rho_2(\overline t-\overline y/2,\overline t+\overline y/2)$ over the variables $\overline t$
taking into account Lemmas 1 and 3. According to the Plancherel equality we get
\begin{equation}\label{intT2}
P_{12}=\int\limits_{\mathbb{R}^{2n}}f_{\rho_1}(\overline x,\overline y)f_{\rho_2}(-\overline x,-\overline y)
d^n\overline xd^n\overline y,
\end{equation}
where $f_{\rho_1}(\overline x,\overline y)$ and $f_{\rho_2}(-\overline x,-\overline y)$
are defined by (\ref{Baker}).

The change of variables $x_j=\lambda_j\cos\alpha_j$, $y_j=\lambda_j\sin\alpha_j$
in the integral (\ref{intT2}) gives rise to the following relation:
\begin{equation*}
P_{12}=\int\limits_{[0;\,\pi]^n\times\mathbb{R}^{n}}d^n\overline\alpha d^n\overline\lambda\,
\bigg(\prod_{j=1}^n|\lambda_j|\bigg)
f_{\rho_1}(\overline{\lambda\cos\alpha},\overline{\lambda\sin\alpha})
f_{\rho_2}(-\overline{\lambda\cos\alpha},-\overline{\lambda\sin\alpha}).
\end{equation*}
Due to Lemma 2 the functions $f_{\rho_1}(\overline{\lambda\cos\alpha},\overline{\lambda\sin\alpha})$,
$f_{\rho_2}(\overline{\lambda\cos\alpha},\overline{\lambda\sin\alpha})$ belong to  $L_1(\mathbb R^n)$
on $\overline\lambda$, and due to Theorem 1 the functions
$\omega _{\rho_1}(\overline x,\overline\alpha)$, $\omega _{\rho_2}(\overline{x},\overline\alpha)$
belong to  $L_1(\mathbb R^n)$ on $\overline x$.
Applying formula (\ref{omegatof}), which is the inverse Fourier transform of (\ref{tomogram}),
we obtain the result (\ref{equalityt2}).
\selectlanguage{english}

\end{proof}

\section {Fractional Fourier transforms of quantum states}

Using the Hamiltonians $\hat H_j=\frac {\hat p_j^2+\hat q_j^2}{2}$ let us define a unitary representation
of the $n$th powers of the circle group ${\mathbb T}=[0,2\pi ]^n$ by the formula
$$
{\mathbb T}\ni \overline \alpha \to \hat U_{\overline \alpha }=\exp\left(-i\overline\alpha\cdot\overline{\hat H}\,\right).
$$
It is straightforward to check that
\begin{equation}\label{rot}
\hat U_{\overline \alpha }^\dagger\hat q_j\hat U_{\overline \alpha }=\hat q_j\cos\alpha_j + \hat p_j\sin\alpha_j.
\end{equation}
Given $\alpha \in (0;\pi)$ let us consider the unitary operator $\tilde{{\mathcal F}}_{\alpha}$ known as
the fractional Fourier transform \cite {Namias} and determined by
the formula
\begin{equation}\label{fract}
\tilde{{\mathcal F}}_{\alpha }^q[\varphi(q)](x)=\frac {\exp\left({\frac {i x^2\cos\alpha}{2\sin\alpha }}\right)}{\sqrt {2\pi |\sin \alpha |}}
\int \limits _{\mathbb R}
\exp\left(\frac {iq^2\cos\alpha}{2\sin\alpha }-\frac {ixq}{\sin\alpha }\right)\varphi(q)dq.
\end{equation}
Using (\ref {fract}) we get \cite {Namias}
\begin{equation*}
\left |\hat U_{\overline \alpha }[\psi(\overline q)](\overline x)\right |^2=
\left |\tilde{{\mathcal F}}_{\alpha _1}^{q_1}\circ \dots \circ \tilde{{\mathcal F}}_{\alpha _n}^{q_n}[\psi(\overline q)]
(\overline x)\right |^2,
\end{equation*}
where $\tilde{\mathcal F}_{\alpha _j}^{q_j}$ acts on functions of the variable $q_j$.

\begin{theorem}\label{3} Suppose that $\rho (\overline x,\overline y)=\psi (\overline x)\psi ^*(\overline y)$, where
$\psi \in {\mathbb V}({\mathbb R}^n)$. Then,
\begin{equation*}
\omega _{\rho }(\overline x,\overline \alpha )=\left |\hat U_{\overline \alpha }[\psi](\overline x)\right |^2.
\end{equation*}
\end{theorem}

\begin{proof}

The space ${\mathbb V}({\mathbb R}^n)$ can be alternatively described as $f\in {\mathbb V}({\mathbb R}^n)$ iff
$$
\hat q_1^{s_1}q_2^{s_2}\cdots \hat q_n^{s_n}f,\ \hat p_1^{s_1}\hat p_2^{s_2}\cdots \hat p_n^{s_n}f\in L_2({\mathbb R}^n),
$$
wherever integer powers $s_j$,
$\sum \limits _{j=1}^ns_j\le n+1$.
Hence, the property  (\ref {rot}) results in $\hat U_{\overline \alpha }{\mathbb V}({\mathbb R}^n)={\mathbb V}({\mathbb R}^n)$,
and we obtain
$
\psi _{\overline \alpha }\equiv \hat U_{\overline \alpha }\psi \in {\mathbb V}({\mathbb R}^n),
$
where $\psi _{\overline \alpha }$ is a wave function of the state $\hat \rho $ in the representation of
\mbox{$\overline{\hat x}(\overline\alpha)=\overline {\hat q\cos \alpha }+\overline {\hat p\sin\alpha }$}.
The condition $\psi \in {\mathbb V}({\mathbb R}^n)$ is equivalent to
$\rho \in \mathcal{V}({\mathbb R}^{2n})$.
Taking into account that $\omega _{\rho}(\overline x,\overline \alpha )$ is a probability distribution corresponding to the observable
$\overline{\hat x}(\overline\alpha)$ in the state $\hat \rho $, we get the result.
Note that $\omega _{\rho }(\cdot ,\overline \alpha )\in L_1({\mathbb R}^n)\cap C({\mathbb R}^n)$ according to Theorem 1.
\sloppy

\end{proof}

{\bf Remark.} {\it Analogously, given a density matrix of mixed state in the position representation
$\rho(\overline q,\overline q')\in\mathcal{V}(\mathbb{R}^{2n})$
the corresponding density matrix $\rho_{\overline\alpha}(\overline x,\overline x')$
in the representation of $\overline{\hat x}(\overline\alpha)$ equals

\begin{equation*}
\rho_{\overline\alpha}(\overline x,\overline x')=\left\langle\overline x\left|\hat U_{\overline\alpha}
\hat\rho\hat U_{\overline\alpha}^\dagger\right|\overline x'\right\rangle=
\tilde{{\mathcal F}}_{\overline\alpha}^{\overline q}(\overline x)\circ
\tilde{{\mathcal F}}_{-\overline\alpha}^{\overline q'}(\overline x')[\rho(\overline q,\overline q')].
\end{equation*}
So, $\rho_{\overline\alpha}(\overline x,\overline x')\in\mathcal{V}(\mathbb{R}^{2n})$ according to Lemma 1.

The tomogram $\omega _{\rho}(\overline x,\overline \alpha )$ is a probability distribution of
observable $\overline{\hat x}(\overline\alpha)$ in the state $\hat\rho$, then
$\omega _{\rho}(\overline x,\overline \alpha )$ is the diagonal matrix element of $\rho_{\overline\alpha}(\overline x,\overline x')$,
i.e. it is the restriction of $\rho_{\overline\alpha}(\overline x,\overline x')$ to the hyperplane $\overline x'=\overline x$
\begin{equation}            \label{omegafractrho}
\omega _{\rho}(\overline x,\overline \alpha )=\rho_{\overline\alpha}(\overline x,\overline x')
\big|_{\overline x'=\overline x}~~.
\end{equation}
Moreover,
$\omega _{\rho}(\cdot ,\overline \alpha )\in C({\mathbb R}^n)\cap L_1({\mathbb R}^n)$,
and formulas (\ref{omegafractrho}), (\ref{tomogram}), (\ref{tomogram2}) define the same object called the optical quantum tomogram.} $\Box $

\section{Conclusion}

Our goal was to search for maximally common conditions on kernels of integral operators under which
Wigner functions and optical quantum tomograms as well as linking their formulas are correctly defined.
We have obtained sufficient conditions of this kind using the Sobolev Embedding theorem.

We defined the space
$\mathcal{V}(\mathbb{R}^{2n})=W_2^{n+1}({\mathbb R}^{2n})\cap {\mathcal F}\left[W_2^{n+1}({\mathbb R}^{2n})\right]$
of kernels $\rho(\overline q,\overline q')$
for which the continuous and integrable tomograms $\omega _{\rho}(\overline x,\overline \alpha )$
exist as well as the inverse transformations of these tomograms into original kernels.
The space $\mathcal{V}(\mathbb{R}^{2n})$
is a subspace of the Sobolev space $W_2^{n+1}(\mathbb{R}^{2n})$ invariant with respect to  to the Fourier transform,
where $W_2^{n+1}(\mathbb{R}^{2n})$ consists of the functions belonging to the space $L_2(\mathbb{R}^{2n})$,
having generalized derivatives of the order of $n+1$  such that
$|\overline x|^{n+1}{\mathcal F}[\rho(\overline q,\overline q')](\overline x)\in L_2(\mathbb{R}^{2n})$.
The space $\mathcal{V}(\mathbb{R}^{2n})$ is narrower than $L_2(\mathbb{R}^{2n})$ or
$W_2^{n+1}(\mathbb{R}^{2n})$, but it is much wider than the Schwarz space $\mathcal{S}(\mathbb{R}^{2n})$.
The main advantage of our approach is a justification of correctness for the integral formula linking optical quantum tomograms
and Wigner functions by means of the Radon transform.


\end{document}